\documentclass[aip,jmp,author-numerical,preprint]{revtex4-1}
\usepackage{amsmath}
\usepackage{amsthm}
\usepackage{bm}
\usepackage{amssymb}
\usepackage{stmaryrd}
\usepackage{mathtools}
\usepackage{enumitem}
\newcommand{\nicefrac}[2]{\ ^{#1}/_{#2}}

\usepackage{wasysym}

\usepackage{color}

\usepackage[normalem]{ulem}

\newtheorem{theorem}{Theorem}[section]
\newtheorem{lemma}[theorem]{Lemma}
\newtheorem{proposition}[theorem]{Proposition}
\newtheorem{corollary}[theorem]{Corollary}

\theoremstyle{definition}
\newtheorem{definition}{Definition}
\newtheorem*{remarks}{Remarks}
\newtheorem*{remark}{Remark}

\newcommand{\bt}{\begin{theorem}}
\newcommand{\et}{\end{theorem}}
\newcommand{\bl}{\begin{lemma}}
\newcommand{\el}{\end{lemma}}
\newcommand{\bc}{\begin{corollary}}
\newcommand{\ec}{\end{corollary}}
\newcommand{\bp}{\begin{proposition}}
\newcommand{\ep}{\end{proposition}}
\newcommand{\bpf}{\begin{proof}}
\newcommand{\epf}{\end{proof}}
\newcommand{\be}{\begin{equation}} 
\newcommand{\ee}{\end{equation}}
\newcommand{\bal}{\begin{align}} 
\newcommand{\eal}{\end{align}}
\newcommand{\beq}{\begin{eqnarray}}
\newcommand{\eeq}{\end{eqnarray}}
\newcommand{\ba}{\begin{array}}
\newcommand{\ea}{\end{array}}
\newcommand{\bma}{\begin{bmatrix}}
\newcommand{\ema}{\end{bmatrix}}
\newcommand{\bi}{\begin{itemize}}
\newcommand{\ei}{\end{itemize}}

\newcommand{\comm}[1]{}

\makeatletter
\newcommand{\opnorm}{\@ifstar\@opnorms\@opnorm}
\newcommand{\@opnorms}[1]{%
  \left|\mkern-1.5mu\left|\mkern-1.5mu\left|
   #1
  \right|\mkern-1.5mu\right|\mkern-1.5mu\right|
}
\newcommand{\@opnorm}[2][]{%
  \mathopen{#1|\mkern-1.5mu#1|\mkern-1.5mu#1|}
  #2
  \mathclose{#1|\mkern-1.5mu#1|\mkern-1.5mu#1|}
}
\makeatother

\newcommand{\ph}{\phi}

\newcommand{\e}{\rm e}

\newcommand{\cH}{{\mathcal H}}

\newcommand{\cK}{{\mathcal K}}

\newcommand{\bbC}{{\mathbb C}}

\newcommand{\bbE}{{\mathbb E}}

\newcommand{\bbN}{{\mathbb N}}

\newcommand{\bbP}{{\mathbb P}}

\newcommand{\bbR}{{\mathbb R}}

\newcommand{\bbZ}{{\mathbb Z}}

\newcommand{\fh}{{\mathfrak h}}

\newcommand{\fD}{{\mathfrak D}}

\newcommand{\tr}{  \textrm{tr\ }  }

\newcommand{\wt}{\widetilde}

\newcommand{\wh}{\widehat}

\newcommand{\Ev}[1]{\E \left( #1 \right)}  

\newcommand{\norm}[1]{\left\Vert#1\right\Vert}
\newcommand{\abs}[1]{\left\vert#1\right\vert}

\newcommand{\setb}[2]{\left \{ #1 \ \middle | \ #2 \right \} }

\newcommand{\bra}[1]{\left < #1 \right |}
\newcommand{\ket}[1]{\left | #1 \right >}
\newcommand{\dirac}[3]{\bra{#1} #2 \ket{#3}} 
\newcommand{\diracip}[2]{\left <#1 \middle | #2 \right >}

\newcommand{\bb}[1]{\mathbb{#1}}
\newcommand{\mc}[1]{\mathcal{#1}}

\def\Z{\mathbb Z}

\def\E{\mathbb E}

\def\e{\mathrm e}
\def\im{\mathrm i}

\def\half {\frac{1}{2}}
\def\1{{\mathsf 1}}
\def\tem{\textemdash}

\def\tr{\operatorname{tr}}    
\def\supp{\operatorname{supp}} 

\def\ra{\rightarrow}

\usepackage{bbm}
\def\bb1{\mathbbm{1}}

\newcommand{\tn}{\textnormal}

\DeclareMathOperator{\rng}{Range}

\begin{document}
\title{Resonant Tunneling in a System with Correlated Pure Point Spectrum}
\author{Rajinder Mavi} 
\email{mavi.maths@gmail.com}
\thanks{Supported by a post-doctoral fellowship from the MSU Institute for Mathematical and Theoretical Physics}
\author{Jeffrey Schenker}
\email{schenke6@msu.edu}
\thanks{Supported by the National Science Foundation under Grant No. 1500386.}
\affiliation{Michigan State University, Department of Mathematics \\  Wells Hall \\ 619 Red Cedar Road \\ East Lansing, MI 48823}
\date{October, 2017; revised October 2018} 

\begin{abstract}We consider resonant tunneling between disorder localized states in a potential energy displaying perfect correlations over large distances.  The phenomenon described here may be of relevance to models exhibiting many-body localization.  Furthermore, in the context of single particle operators, our examples demonstrate that exponential resolvent localization does not imply exponential dynamical localization for random Schr\"odinger operators with correlated potentials.	
\end{abstract}

\maketitle 
 
\section{Introduction}
In this note we consider the nature of Anderson localization in a quantum system with a random potential energy displaying strong correlations over large distances. In multi- and many-body systems, such correlations are typical when the potential is viewed as a function on a Fock space basis, with the metric geometry induced by the allowed hopping terms.  A trivial example of such correlations can be found in a system of identical particles, as states related by exchange of particles have identical potential energy, regardless of the distance between the exchanged particles.  Of course in such situations one is typically interested in indistinguishable particles satisfying Bose or Fermi statistics, for which such exchange does not produce a new state.  Nonetheless, other instances of particle rearrangement may lead to correlations that are physically relevant.  For instance, in recent work of the authors, a system consisting of a tracer particle and a field of Harmonic oscillators is considered.\cite{MS2016} In that work, the potential energy depends on the particle position and the total excitation number in the oscillator fields.  Two physically distinguishable configurations with the same particle position may have identical potential energy even if they have substantially different oscillator configurations.

Our objective here is to consider a simple example of tunneling that may be induced by correlations of this type.  The model system we consider consists of a single spin-$\nicefrac{1}{2}$ tight binding particle on $\bbZ^d$ subject to a random potential energy, weak nearest neighbor hopping and a local term that can flip the spin when the particle visits the origin.  It follows from known results in the literature that this system exhibits Anderson localization, i.e.,  pure point spectrum with eigenfunctions that decay exponentially with respect to position.  However eigenstates, and therefore particle dynamics, are \emph{not} localized in spin.  Indeed, we show below that there are states in which the particle remains well localized at a distance $L$ from the origin as its spin oscillates between up and down with a frequency of order $\e^{\mu L}$.
   
The basic physical reason behind this oscillation is quite simple.  The exact eigenfunctions of the model are symmetric and anti-symmetric with respect to spin flip. Turning on the spin flip term at the origin splits a doubly degenerate eigenvalue associated to eigenfunctions localized at distance $L$ from the origin to a pair of eigenvalues with energy gap $\approx \e^{-\mu L}$.  The size of the splitting is governed by the amplitude of the corresponding eigenvector at the origin, and is therefore substantially suppressed by Anderson localization, but is not typically zero.  Despite the simplicity of this idea, implementing it in the presence of a background of dense point spectrum requires a bit of technical effort. 

Our main interest in this model lies in the fact that the system Green's function, consisting of matrix elements of the resolvent, exhibits stronger decay than the eigenfunctions.  Indeed, the position/spin basis for the Hilbert space of this system is naturally associated to a graph $\Gamma$ consisting of two copies of $\bbZ^d$ with an edge connecting the origins of both lattices. We label a site on this graph by $(x,i)$ with $x\in \bbZ^d$ and $i=\pm 1$, corresponding to spin up and spin down, and denote the corresponding state in Dirac notation by $\ket{x,i}$.  The \emph{graph metric}  on $\Gamma$ is 
\begin{equation}  
	d_\Gamma(x,i;y,j) \ = \ \left\{\begin{matrix} |x-y| &\tn{ if } i = j\\
  1+|x| + |y|  &\tn{ if } i\neq j  \end{matrix}\right. . 
\end{equation}
Here and throughout $|x|$ denotes the $\ell^1$ norm of a vector in $\bbZ^d$, $|(x_1,\ldots,x_d)|=\sum_j |x_j|$.
For example, $d_\Gamma(x,1;x,-1) = 1 + 2 |x|$.  We show below that fractional moments of the Green's function $ G_z(x,i;y,j)  =  \dirac{x,i}{\left (\fh - z \right )^{-1}}{y,j} , $
 with $\fh$ the system hamiltonian, exhibit exponential decay in the metric $d_\Gamma$.  Specifically, we obtain
 \begin{equation}\label{eq:resolvedecay}
 \Ev{ \abs{G_z(x,i;y,j)}^s} \ \le \ A_s \e^{-\mu_s d_\Gamma(x,i ;y,j)}
 \end{equation}
 for $0 < s <1$,
with constants $A_s,\mu_s$ that are independent of $(x,i), (y,j)$ and $z$.  By way of contrast the dynamical correlator $\sup_{t\in \bbR} \abs{\dirac{x,i}{\e^{-\im t \fh}}{y,j}}$ cannot exhibit decay in the metric $d_\Gamma$, as this would contradict the existence of the oscillating spin states we construct.  Instead the dynamical correlator satisfies a localization bound of the form
\begin{equation}
\Ev{\sup_{t \in \bbR} \abs{\dirac{x,i}{\e^{-\im t \fh}}{y,j}}} \ \le \ \e^{-\mu | x - y|}.
\end{equation}
Thus $\sup_{t \in \bbR} \abs{\dirac{x,+1}{\e^{-\im t H}}{x,-1}}$ is of order one, although $\abs{G_z(x,+1;x,-1)}$ is of order $\e^{-\mu |x|}$.
 
We now turn to the specific description of our model. The system Hamiltonian we consider is the following operator on $\ell^2(\Z^d;\bbC^2)$
	\be 
  \fh_g \ = \ H \otimes \mathbbm{1} +g \ket{0}\bra{0}\otimes \sigma^{(1)} 
\ee
where $\mathbbm{1} = \bigl ( \begin{smallmatrix}
	1 & 0 \\ 0 & 1 
\end{smallmatrix} \bigr )$, $\sigma^{(1)}= \bigl ( \begin{smallmatrix}0 & 1\\ 1 & 0
\end{smallmatrix} \bigr )$, and $H$ denotes the Anderson Hamiltonian on $\ell^2(\Z^d)$:
	\begin{equation}\label{anderson}      
		H \ = \ \gamma \Delta + V \ ,
	\end{equation}
with $\Delta$  the discrete Laplacian and $V$ a diagonal matrix with independent, identically distributed random entries distributed according to a measure possessing a bounded density.   Written out in the basis $\Gamma$, we have
\begin{multline} 
\label{andersondoublelattice}
  \fh_g \ = \ \gamma \sum_{\substack{x,y\in \bbZ^d, \\ \|x-y\|_1 = 1 }} \sum_{i=\pm 1}   \ket{x,i}\bra{y,i} \ + \ \sum_{x\in \bbZ^d} \sum_{i=\pm 1} V_x \ket{x,i}\bra{x,i} \\ + \ g ( \ket{0,+1}\bra{0,-1} + \ket{0,-1}\bra{0,1} ).
\end{multline}

It is known that, for sufficiently small $\gamma$, with probability one the Anderson hamiltonian $H$ has pure point spectrum with exponentially localized eigenvectors satisfying the SULE condition.\cite{DelRio1996} That is, each eigenvector $\ket{\psi_j}$ has a \emph{localization center} $\xi_j \in \bbZ^d$ such that 
\be \abs{\diracip{x}{\psi_j}} \ \le \ A_\omega (1+|x|)^\nu \e^{-\mu |x-\xi_j|}, \ee
with $\mu >0$ and $\nu >d$ non-random and $A_\omega$ a random pre-factor that is finite with probability one. Furthermore, the eigenvalues of the Anderson hamiltonian (\ref{anderson}) are simple in the localization regime\cite{Simon1994,JL2000} and, in a suitable scaling limit, the joint process $(\xi_j, E_j)$ of localization centers and eigenvalues becomes a Poisson process.\cite{Killip2007}

The system Hamiltonian \eqref{andersondoublelattice} commutes with the Pauli matrix $\sigma^{(1)}$, which flips the particle spin.  Thus it is possible to choose a basis of eigenstates for $\fh_g$ consisting of functions symmetric or anti-symmetric with respect to spin flip. However, we show below that with probability one the eigenvalues of $\fh_g$ are simple for $g\neq 0$.  Thus, for $g\neq 0$ every eigenfunction is either symmetric or anti-symmetric.  In particular any eigenvector $\ket{\Psi}$ of $\fh_g$ for $g\neq 0$ satisfies
\be \abs{\diracip{x,+1}{\Psi}} = \abs{\diracip{x,-1}{\Psi}} \label{eq:Psiequal},\ee
and does not decay in the metric $d_\Gamma$.  We emphasize that although the Green's function between $\ket{x,+1}$ and $\ket{x,-1}$ is exponentially small,  this decay is not inherited by the eigenfunctions, all of which satisfy eq.\ \eqref{eq:Psiequal}.  Thus (\ref{andersondoublelattice}) provides an example of Green's function decay without corresponding eigenfunction localization.

For $g=0$, the spectrum is doubly degenerate.  Corresponding to an eigenvector $\ket{\psi_j}$ of the Anderson hamiltonian \eqref{anderson} with eigenvalue $E_j$ we have the following symmetric and anti-symmetric eigenvectors of $\fh_0$, with eigenvalue $E_j$:
\be \ket{\psi_j}_+ \ = \ \frac{1}{\sqrt{2}}\left (\ket{\psi_j,+1} + \ket{\psi_j,-1}\right ) \quad \text{and} \quad \ket{\psi_j}_- \ = \ \frac{1}{\sqrt{2}}\left (\ket{\psi_j,+1} -\ket{\psi_j,-1}\right ).\ee
Roughly speaking, we expect for $g\neq 0$ not too large that there correspond to $\psi_j$ two eigenvectors
$ \ket{\psi_j}_\pm^g $ with eigenvalues $E_j^\pm(g)$ satisfying $\abs{E_j^+(g) - E_j^-(g)} \ \le \ \abs{g}^2 \e^{-\mu |\xi_j|}$ where $\xi_j$ is the localization center of $\ket{\psi_j}$.  Furthermore, the eigenvectors $\ket{\psi_j}_\pm^g$ are both centered around $\xi_j$ in the sense that
\be \abs{\diracip{x,\sigma }{\psi_j}_\pm^g} \ \le \ C_\omega (1+|x|)^\nu \e^{-\mu |x-\xi_j|}. \ee
The split levels $\ket{\psi_j}_\pm^g$ form a closely correlated pair which leads to non-trivial dynamics over the time scale $\abs{E_j^+(g) - E_j^-(g)}^{-1}$ for a state prepared in the eigenstate $\ket{\psi_j,+1}$ of the $g = 0$ system.  Our main result is that these properties do indeed hold with probability one, at least for a subsequence of eigenvectors in $\fh_g$.
      
\subsection{Relation to multi-particle models}  
We would like to draw attention to a similarity between the Hamiltonian (\ref{andersondoublelattice}) and the multi-particle Anderson model for \emph{distinguishable particles}. 
We present an analysis of the dynamics of  (\ref{andersondoublelattice}) and we argue that both models exhibit related resonant tunneling behavior. In the multi-particle Anderson model, no such results are rigorously established although the possibility for such tunneling was mentioned in ref.\ \onlinecite{AW2009}.  

The multi-particle Anderson hamiltonian $H^{(N)}$ acts the Hilbert space $\otimes_{i = 1}^N \ell^{2}(\bbZ^d)$ and is given by a sum $H^{(N)} = H_0^{(N)}+ U$, where $H_0^{(N)}=\sum_{i=1}^N H_i$  where $H_i$ denotes the Anderson hamiltonian acting on the $i^{\text{th}}$ particle (corresponding to the $i^{\text{th}}$ factor in $\otimes_{i = 1}^N \ell^{2}(\bbZ^d)$) and $U$ is a symmetric short range interaction among the particles. This operator is symmetric under any permutation of the particles.  Generally, one is interested in a system of identical Bosons or Fermions, which corresponds to restricting $H^{(N)}$ to the symmetric or anti-symmetric subspace of $\otimes_{i = 1}^N \ell^{2}(\bbZ^d)$.  In the weak hopping regime, fractional moment bounds and dynamical localization in terms of a symmetrized metric proper for such indistinguishable particles were obtained in ref.\ \onlinecite{AW2009}. 
 
The phenomenon addressed in this paper, is analogous however to localization for a collection of {\it distinguishable} particles.  For example consider the two particle system at weak hopping.  Consider an initial state $\Psi_0$ which is localized near $(x_0,y_0)$ with $x_0$ and $y_0$ far from each other, corresponding roughly to particles localized at the distant points $x_0$ and $y_0$.  For instance, we could take $\Psi_0(x,y)=\psi_1(x)\psi_2(x)$ with $\psi_1$ and $\psi_2$ eigenstates of the Anderson hamiltonian with localization centers at $x_0$ and $y_0$ respectively. Analogous to the transitions we demonstrate for (\ref{andersondoublelattice}), we expect $\Psi_\tau$ to transition at some large time $\tau$ to a state localized near $(y_0,x_0)$.  Furthermore, as shown in ref.\ \onlinecite{AW2009}, the state will remain localized near $\{(x_0,y_0),(y_0,x_0)\}$ for all time.  Similar transitions are to be expected for systems with $N >2$.  To our knowledge, the nature of such transitions in the multi-particle case has not been studied in the literature. For indistinguishable particles, this analysis suggests that in localized systems it may be hard to distinguish Bosons from Fermions over exponentially long scales.

In a many body setting, one may expect analogous transitions of excitations to occur in infinite dimensional spaces.
In ref.\ \onlinecite{MS2016}, we consider the localized phase of a model  of a tracer particle interacting with a field of oscillators. 
The oscillators have a uniform frequency. Thus without the tracer particle there is infinite degeneracy of the energy levels. Once the tracer particle is added these energy levels split, and we prove a localization result somewhat analogous to known results for the  multi-particle Anderson model. We prove dynamical localization for the tracer particle up to any finite energy, but we are not able to control the  dynamics of the excitations among the oscillators.  The point of the present work is that this may not be a technical issue, since it is possible the oscillators can undergo resonant tunneling mediated by the localized particle.
          
 Our goal in this paper is to develop an understanding for tunneling among nearly degenerate states in the relatively simple model (\ref{andersondoublelattice}). We emphasize that similar dynamical properties in more complicated models, such as those studied in Refs.\ \onlinecite{AW2009,MS2016}, are completely unstudied in the literature.

  \section{Main Results} 
  We formulate our results for a slightly more general version of the operator $\fh_g$, taking
   \begin{equation}   \label{andersonspin}  
                   \fh_g  \ = \ (\gamma \Delta + V  ) \otimes \mathbbm{1}   
                                         +  g  \ket{\zeta}\bra{\zeta} \otimes \sigma^{(1)}
   \end{equation}
    where $\zeta \in \ell^2(\bbZ^d)$ is a fixed function with bounded support. Without loss of generality we assumes $\norm{\zeta}_2 = 1$, since this amounts to a normalization of the coupling $g$. Given $\phi\in \ell^2(\bbZ^d)$ and $i=\pm 1$, we write $\ket{\phi,i}$ for the pure spin state $\sum_x \phi(x) \ket{x,i} = \ket{\phi} \otimes \ket{i}$ with positional wave function $\phi$.  We use the centered Laplacian  
    \[ \dirac{ x,i} {\Delta\otimes \bb1} {y,j} \ = \ \delta_{|x-y|,1}\delta_{i,j}. \]
  The potential $V\otimes \bb1$ is given by 
  \[  \dirac{x}{V\otimes \bb1}{y} \ = \ V_x \delta_{x,y} \delta_{i,j}    \]
  where the random variables  $\{V_x\}_{x \in \bbZ^d}$ are independent and identically distributed.  We further assume that the distribution of $V_x$ has a bounded density $\rho$. 
 We write $\bbP$ for the product distribution of the potential process and $\bbE$ for expectations with respect to $ \bbP $.
 
The starting point for our analysis is the following well known fractional moment bound for the Anderson hamiltonian:
\begin{theorem}[Aizenman and Molchanov\cite{AM94}]\label{AMThm}There is $\gamma_0 >0$ so that for any $|\gamma| < \gamma_0$, $z\in \bbC^+$ and $s <1$     
\begin{equation} \label{fm}
           \bbE\left( \left|\dirac{x}{ (H - z )^{-1}}{ y} \right|^s   \right)   
                                  \ < \ A_{\gamma,s} \e^{-\mu_{\gamma,s}|x - y|} ,
   \end{equation}  
with $\mu_{\gamma,s} >0$ and $A_{\gamma,s} < \infty$. 
\end{theorem}
See also ref.\ \onlinecite{AW2015}, Theorem 6.3, where  an explicit bound is given on the right hand side. In both refs.\ \onlinecite{AM94,AW2015}, the theorem is stated with $s$ fixed, however once eq.\ \eqref{fm} holds for $s=s_0$ it can be extended to all $0<s<1$ by an ``all-for-one'' Lemma such as Lem.\ B.2 of ref.\ \onlinecite{Aizenman2001}. 
        
Throughout this paper $\gamma_0$ denotes the critical hopping for the Anderson hamiltonian appearing in Thm.\ \ref{AMThm}.  For $|\gamma|<\gamma_0$, fractional moment bounds for  $\fh_g$ may be inferred almost immediately from  eq.\ (\ref{fm}).
    \begin{theorem} 
     \label{fractional moment} For $|\gamma| <\gamma_0$, $0 < s < 1$, $g\in \bbR$, and $z\in \bbC^+$, 
       \be \label{fm2}      \bbE\left( \left |\dirac{ x,i}{ (\fh_g - z )^{-1} }{ y,j } \right|^s   \right)   \   <  \
                    \wt{A}_{\gamma,s,g}  \e^{-\wt{\mu}_{\gamma,s} d_\Gamma(x,i;y,j) }  ,
      \ee
      with $\wt{A}_{\gamma,s,g} <\infty $ and $\mu_{\gamma,s} >0$ and independent of $g$. 
    \end{theorem}
    For $s<\nicefrac{1}{3}$, Thm.\ \ref{fractional moment} follows from Thm.\ \ref{AMThm} by a straightforward application of the resolvent expansion, H\"older's inequality and boundedness of fractional moments.  It can then be extended to $\nicefrac{1}{3} \le s < 1$ by an ``all-for-one'' Lemma. For completeness we give the proof in Appendix \ref{app1} below.

One virtue of fractional moment bounds such as eq.\ (\ref{fm}) has been their use in deriving various precise
      statements of Anderson localization. In particular, dynamical localization for the Anderson hamiltonian (\ref{anderson}) was first obtained by this method.\cite{Aizenman1994}  Here ``dynamical localization'' refers to a bound of the form
    \begin{equation}
     \label{dynamical localization}
       \bbE \left(  \sup_t   | \langle y |  \e^{- \im t H} | x \rangle | \right) \ < \ C \e^{ - \nu |x - y|},
   \end{equation}
   with $C < \infty$ and $\nu > 0$. This is a strong form of Anderson Localization, which implies spectral localization (pure point spectrum) and  bounds on eigenfunction correlators, see ref.\ \onlinecite{AW2015}, Chapter 7.

Despite the bound eq.\ \eqref{fm2}, no such strong form of Anderson localization holds for $\fh_g$ in the metric $d_\Gamma$. For arbitrarily long times, the best we can do is  obtain a bound that only gives decay in the $\bbZ^d$ directions:
\begin{corollary}
For $|\gamma| <\gamma_0$ and $g\in \bbR$ there are $C$ and $\nu$ depending on $\gamma$ and $g$ such that
\be  \bbE \left(  \sup_t   | \dirac{x,i}{\e^{- \im t \fh_g}}{y,j} | \right) \ < \ C \e^{ - \nu |x - y|}.\ee
\end{corollary}
\noindent However, we can go beyond this and extract a bound on rate of spin flips  for bounded times:
    \begin{corollary}
     \label{fm dynamical bound} For $|\gamma| <\gamma_0$ and $g\in \bbR$ there are $C$ and $\nu$ depending on $\gamma$ and $g$ such that for $i\neq j$
       \be \label{fmdbeqn}     \bbE\left( \sup_{t} \frac{1}{|t|}  \left|\dirac{x,i} {\e^{ - \im t \fh_g}}{  y,j}  \right|    \right)   
              \ < \ C \e^{ - \nu (|x|+|y|)} . \ee
   \end{corollary} 
\noindent In particular, if the particle is initially placed in the state $\ket{x,+1}$, then with probability one
   \be \sum_{y}  \abs{\dirac{y,-1}{\e^{-\im t \fh_g}}{x,1}}^2   \ \le \ C_\omega |t|^2 \e^{-2\nu |x|}, \ee 
   with $C_\omega <\infty$ almost surely.  Thus the time for the spin to flip from $+1$ to $-1$ is bounded below by $\e^{\nu |x|}$.
    
Bounds of the form eq.\ \eqref{fmdbeqn} may be practical for studying split level resonances
    in a variety of Anderson localized systems as they would be easy to obtain in many cases. 
     Indeed, such a bound can be obtained for the model introduced in ref.\ \onlinecite{MS2016}.
On the other hand, for $\fh_g$ we can study the dynamics between split energy levels much more precisely.

Our main theorem is the following.
    Let us define the localization center of a vector $\phi \in \ell^2(\Gamma)$ as a point $x \in \bbZ^d$ so that
     $ |\diracip{ x,-1}{ \phi } | \vee   |\diracip{ x, 1}{ \phi } | $ is maximized. Let $\wh P_x$ denote the projection of $\ell^2(\Gamma)$ onto states with $\bbZ^d$ coordinate within  $ \nicefrac{|x|}{2}$ of $x$,  that is 
          \be \label{spinprojection}
          \wh P_x \ = \ \sum_{ u,j : |u-x| < \nicefrac{|x|}{2}  }  \ket{u , j}\bra{ u , j}   .
          \ee
   We will examine the dynamics of a nearly localized eigenvector with  spin $i=1$.

    \begin{theorem}
     \label{main1}  Suppose that $\abs{\gamma}<\gamma_0$. Then there is an $\epsilon >0$ such that for each $g\in \bbR$, with probability one, there is a sequence of distinct vectors $\phi_k^{(g)} \in \ell^2(\bbZ^d)$ with localization centers $x_k^{(g)} \ra \infty$ and associated energies $\lambda_k^{(g)}$ with the following properties:
     \begin{enumerate}
     \item 	 There is a time scale $\tau_k > \e^{ \epsilon |x_k^{(g)}|}$ so that, for any $n\in \bbZ$ and $t \in \bbR$ satisfying $\left|t - n \tau_k \right| < \e^{ \frac{1}{2} \epsilon  |x_k^{(g)}| }$,  
     \be  \label{main1eq1}
       \left \| \e^{- \im t   \fh_g  }  \ket{  \phi_k^{(g)} ,  1}
          -
        \e^{- \im t \lambda_k^{(g)}} \ket{ \phi_k^{(g)},(-1)^n}\right \|   
             \ < \  \e^{ - \epsilon |x_k^{(g)}|  }
          \ee
\item On the other hand, the support of the wave packet at all times  is almost entirely  contained in  $ \Lambda_{\nicefrac{|x_k^{(g)}|}{2}}(x_k^{(g)}),$
          \be \label{tunnel1}
          \inf_{ t > 0 }   \norm{\wh P_{x_k} \e^{ - \im t \fh_g} \ket{\phi_k^{(g)},  1 } }
                       \ > \ 1 - 
                         \e^{-\epsilon |x_k^{(g)}| }.
           \ee
     \end{enumerate}
              \end{theorem}
    \begin{remarks}~
    \begin{enumerate}
    \item The point of eq.\ (\ref{tunnel1}) is that the particle spin flips only by tunneling, as no significant portion of the wave packet is ever near the region connecting the two spins.
    \item In the proof we will choose the vector $\phi_k^{(g)}$ so that $\ket{\phi_k^{(g)},1}+\ket{\phi_k^{(g)},-1}$ is an eigenfunction of $\fh_g$ with eigenvalue $\lambda_k^{(g)}$.  The time scale $\tau_k$ is given by $\pi (\lambda_k^{(g)}-\wt{\lambda}_k^{(g)})^{-1}$ where $\wt{\lambda}_k^{(g)}$ is a nearby eigenvalue associated to an eigenvector that has very large overlap with $\ket{\phi_k^{(g)},1}-\ket{\phi_k^{(g)},-1}$.
    \item Theorem \ref{main1} is proved below, after the statement of Prop.\ \ref{reflect thm} in Section \ref{tunnel}
    \end{enumerate}	
    \end{remarks}

     The rest of the paper is organized as follows. The main Theorem \ref{main1} is proved in Section \ref{r1}, where we show that the model with spin, $\fh_g$,  may be reduced to rank one perturbations of the spinless model, recall the notion of a SULE basis, and state a technical result (Theorem \ref{infinite matching}) that allows us to construct the resonant states described in Theorem \ref{main1}.  The technical Theorem \ref{infinite matching} is proved in Section \ref{minsmatch},  where we recall the Minami estimate and utilize it together with 
           the SULE basis to match vectors after a rank one perturbation.
           In Appendix \ref{app1}  we conclude with short proofs of 
                  Theorem \ref{fractional moment} and Corollary \ref{fm dynamical bound}.

   \section{Proof of the main theorem}\label{r1}
      
\subsection{Rank-one perturbations of $H$} \label{reduction} Let $D=\ket{\zeta}\bra{\zeta}$ and, for $g\in \bbR$, let  
\be \label{r1fam}
            H_g \ := \ \gamma \Delta + V  +  g D,
     \ee
where $H=\gamma \Delta +V$ is the Anderson hamiltonian \eqref{anderson} on $\ell^2(\bbZ^d)$.  In this section we show that $\fh_g \cong H_g \oplus H_{-g}$ and relate spectral properties of $\fh_g$ to spectral properties of $H_{\pm g}$. 

A first observation is that or each $g$ and for  $|\gamma|<\gamma_0$, $H_g$ exhibits Anderson localization, in the sense of having, with probability one, simple pure point spectrum with exponentially decaying eigenfunctions.\cite{JL2000}  Given $g \in \bbR$, let $I^{(g)}$ index the eigenpairs $(\phi_i^{(g)},\lambda_i^{(g)})$
          with normalized eigenvectors $\|\phi_i^{(g)}\|_2 = 1$,
         so that $ (H + gD) \phi^{(g)}_i  = \lambda_i^{(g)} \phi^{(g)}_i  $. 

The eigenvectors and eigenvalues of $\fh_g$ can be determined from the eigenvectors and eigenvalues of $H_{\pm g}$.
        \begin{proposition} \label{hgbasis} Let $\cH_+$ and $\cH_-$ denote the subspaces of $\ell^2(\Gamma)$ symmetric and anti-symmetric with respect to $\sigma^{(1)}$ respectively.  Identify $\cH_\pm$ with $\ell^2(\bbZ^d)$ through the explicit isometries
   \be \label{eq:expisom}\ket{x} \ \mapsto \ \frac{1}{\sqrt{2}} \left ( \ket{x,1} \pm \ket{x,-1} \right ).\ee
   Then $\cH_\pm$ are invariant subspaces for $\fh_g$, and under the identification \eqref{eq:expisom},
   \be\label{eq:Hinv} \left . \fh_g \right |_{\cH_\pm}  \ =  \ H_{\pm g}. \ee  In particular, 
   \begin{enumerate}
   	\item For any $g > 0$ and $\gamma > 0$,  $\sigma(\fh_g) = \sigma(H_{g} ) \cup \sigma(H_{-g})$.
   	\item For $|\gamma| < \gamma_0$, with probability one $\fh_g$ has simple pure point spectrum with eigenvectors exponentially localized in position. Moreover, the eigenvectors of $\fh_g$ are exactly the vectors of the form $ \frac{1}{\sqrt{2}} \left ( \ket{\phi_i^{\pm g},1} \pm \ket{\ph_i^{\pm g},-1} \right ) $ with $i\in I^{(\pm g)}$. 
   \end{enumerate}   
   \end{proposition}

\begin{proof}
    This result essentially follows from  the fact that $\fh_g$  commutes with $\sigma^{(1)}$.   Let $\ket{\pm} = \frac{1}{\sqrt{2}} (\ket{1} \pm \ket{-1} )$.  Then $\sigma^{(1)}\ket{\pm} = \pm \ket{\pm} $ and     
    \be \fh_g \ket{\phi,\pm} \ = \ \fh_g \ket{\phi}\otimes \ket{\pm} \ = \ (H_{\pm g} \ket{\phi}) \otimes \ket{\pm} \ = \ \ket{H_{\pm g} \phi,\pm}\ee
    for any $\phi \in \ell^2(\bbZ^d)$.  Thus $\cH = \cH_+ \oplus \cH_-$ with $\cH_\pm = \ell^2(\bbZ^d) \otimes \ket{\pm}$, eq.\ \eqref{eq:Hinv} holds, and part 1 follows. To prove part 2, recall  that for given $g$ and sufficiently small $\gamma > 0$,
  $H_g $ has almost surely simple point spectrum with exponentially decaying eigenvectors.
    Moreover,  $\sigma_{pp} (H_g) \cap \sigma_{pp} (H_{-g}) = \emptyset$ with probability one.\cite{JL2000}\end{proof}         
         
We now recall the known result that the eigenvectors $\{\phi_i^{(g)}\}_{i\in I^{(g)}}$ form a \emph{Semi Uniformly Localized Eigenfunction} (SULE) basis.\cite{DelRio1996} 
For any sufficiently well localized vector $\phi\in \ell^2(\bbZ^d)$, define the \emph{center of localization} to be a site $x\in \bbZ^d$ at which $ x\mapsto| (1+|x|)^{d+1}\phi(x)|$ attains its maximum value.\cite{Note1} 
For $i\in I^{(g)}$,  let $x_i^{(g)} \in \bbZ^d $ denote the center of  localization  of the eigenfunction $\phi_i^{(g)}$.  
\begin{theorem}[SULE basis\cite{DelRio1996}\tem see also ref.\ \onlinecite{AW2015}, Theorem 7.4] For $g\in \bbR$ and  $0 <\abs{\gamma}< \gamma_0$, there is $\xi = \xi(\gamma,g) > 0$  and an event $\Omega_{\mathrm{SULE}}^{(g)} \subset \Omega$ with $\bbP(\Omega_{\mathrm{SULE}}^{(g)}) = 1$ such that for  $\omega \in \Omega_{\mathrm{SULE}}^{(g)}$ there is $A_\omega <\infty$ so that
          \be\label{suledecay}
                   |\phi_i^{(g)}(x)| \ \leq \ A_\omega (1 + |x_i^{(g)}|)^{d+1}  \e^{ -\xi |x - x_i^{(g)}|}
          \ee
     for every $i\in I^{(g)}$.
     \end{theorem} 
     \begin{remark} Let $\mc{G}\subset \bbR$ be a finite set and let $\Omega_{\mathrm{SULE}}^{(\mc{G})} = \cap_{g\in \mc{G}}\Omega_{\mathrm{SULE}}^{(g)}$. Then $\bbP(\Omega_{\mathrm{SULE}}^{(\mc{G})}) = 1$ and for $\omega \in \Omega_{\mathrm{SULE}}^{(\mc{G})}$, eq.\ \eqref{suledecay} holds for all $g\in \mc{G}$ and $i\in I^{(g)}$.      \end{remark}
     
    SULE implies a limit on the concentration of localization centers. For $\Lambda \subset \bbZ^d$ we define the local index set, 
       \be\label{localindex} 
            I^{(g)}_{\Lambda} \ := \ \left\{  i :  x_i^{(g)}  \in \Lambda  \right\} 
      \ee
     and the corresponding local  spectrum,
     \be\label{localspectrum}
            \Sigma^{(g)}_{\Lambda} \ := \ \left\{  \lambda^{(g)}_i :  i \in   I^{(g)}_{\Lambda}  \right\}.
     \ee
      For $\Lambda \subset \bbZ^{d}$ let us define the projection $P_\Lambda  =   \sum_{x \in \Lambda} |x \rangle \langle x |$ and the local eigenbasis  projection 
     \[     P^{(g)}_\Lambda \ = \ \sum_{i \in I^{(g)}_\Lambda} |\phi^{(g)}_i \rangle \langle \phi^{(g)}_i|.   \]
       In the next section we 
  The following lemma, proved in Section \ref{indices} below, relates the projections  $P_\Lambda$ and $P^{(g)}_\Lambda$.
   \bl  \label{localsum} Let $\xi=\xi(\gamma,g)$ the inverse localization length as above and fix $p>1$ and $\alpha < 1$.  For each $\omega \in \Omega_{\mathrm{SULE}}^{(g)}$ there is a finite $C_\omega = C_\omega(p,\alpha,\xi)$ large enough such that for each $L\geq 1$, $u\in \bbZ^d$ and $\ell \ge  C_\omega + \left (\log (|u|+L) \right )^p$, the following holds. Let $ \Lambda := \Lambda_{L}(u)  $ and   $ \Lambda^{\pm} = \Lambda_{L \pm \ell} (u)$. Then      \be\label{localC}
             \| (1- P_{ \Lambda^{+} } )P_{\Lambda }^{(g)} \|^2 \  <  \    \e^{ -\xi \ell }.   
       \ee
    which implies, 
    \be\label{localub}
              \alpha   \tr P_{\Lambda }^{(g)}    \    <  \ \tr P_{ \Lambda^{+} } .
     \ee
    On the other hand,
      \be\label{localD}
              \|  (1 - P_{\Lambda }^{(g)})  P_{ \Lambda^{-}}\|^2 \ < \   \e^{ - \xi \ell }.   
       \ee
    which implies, 
    \be\label{locallb}
        \alpha    \tr P_{ \Lambda^{-} } \ <  \    \tr P_{\Lambda }^{(g)}     
    \ee
   \el
    \begin{remark} By eq.\ (\ref{localC}),   $ P_{\Lambda}^{(g)} $
     projects almost entirely to
      $ P_{ \Lambda^+}$. Similarly, eq,\ 
      (\ref{localD}) states  that $P_{ \Lambda^-}$
      projects almost entirely to 
      $P_{\Lambda }^{(g)}$.
      These observations lead to the implied statements (\ref{localub}) and (\ref{locallb}), whose proofs are included below. Taking the trace of (\ref{localub}) and (\ref{locallb}), we find  that
  \begin{equation} \label{box concentration}
                                  \alpha    | \Lambda_{L-\ell}(u) |
                        \leq  | I^{(g)}_{\Lambda_{L}(u)  }|  
                        \leq  \alpha^{-1}  | \Lambda_{L+\ell}(u) |.
   \end{equation} 
   \end{remark}

\subsection{Tunneling between corresponding eigenvectors}\label{tunnel}

The key to the proof of Theorem \ref{main1} is to identify eigenfunctions $\ket{\phi_{j_\pm}^{(\pm g)}}$ of $H_{\pm g}$ that have strong overlap, $\abs{\diracip{\phi_{j_+}^{(+g)}}{\phi_{j_-}^{(-g)}}} \approx 1$.  
\begin{definition}   
   We say that an eigenvector $\phi_i^{(g)}$ of  $H +gD$
     $\epsilon$-corresponds to an eigenvector $\phi_j^{(g')}$ of $H + g'D$ if
     \be  \label{vectormatch}  
           | \langle \phi^{(g)}_i , \phi_j^{(g')} \rangle |  \ > \ 1- \epsilon , \ee
    and the corresponding eigenvalues satisfy $|\lambda_i^{(g)} - \lambda_j^{(g')}| < \epsilon$.   
   If furthermore the localization centers $x_i^{(g)}$, $x_j^{(g')}\in \Lambda$ for some box $\Lambda \subset \bbZ^d$, we say that $\phi_i^{(g)}$ and $\phi_j^{(g')}$ are $(\epsilon,\Lambda)$-corresponding. We denote the set of all $(\epsilon,\Lambda)$-corresponding pairs  by $I^{(g,g')}_\Lambda(\epsilon)$.
   For any $\epsilon$-corresponding pair of eigenvectors we may  choose the relative phase so that   
                $ 1-\epsilon \ < \ \langle \phi^{(g)}_i , \phi_j^{(g')} \rangle    \ < \ 1 $.\end{definition}

Our goal is to show for any $g,g'$ that, with probability one, there is a sequence $(\phi_{i_k}^{(g)},\phi_{j_k}^{(g')})$ of $\epsilon_k$-corresponding pairs of eigenvectors with $\epsilon_k\rightarrow 0$. We select sequences $(u_k)_{k= 1}^\infty $ in $\bbZ^d$ and $(L_k)_{k= 1}^\infty $ so that
  \[ \Lambda_k : = \Lambda_{L_k}(u_k)  \]
   are pairwise non intersecting. Any such sequences will do, but for concreteness let $u_0 \in \bbZ^d \setminus \{0\}$ and let $u_k = 3^{k} u_0 $,  $L_k =3^{k - 1} |u_0| $ for $ k \in \bbN$. Note that  
   \be \label{xLbound}  2 L_k \ \le \ |x| \ \le 4 L_k , \quad x\in \Lambda_k.\ee
    By choosing $u_0$ large enough we may assume further that $\supp \zeta \cap \Lambda_k = \emptyset$ for all $k$. 
\begin{theorem}
 \label{infinite matching} Let $(u_k)_k$ and $(L_k)_k$ be as above.  Then there is  $\nu >0$ so that, for any  $g,g'\in \bbR$, there are with probability one infinitely many $k \in \bbN$ for which $I_{\Lambda_k}^{(g,g')}( \e^{ - \nu L_k })$ is non-empty, i.e., there is at least one $(\e^{ - \nu L_k } , \Lambda_k)$-corresponding pair of eigenvectors.
          \end{theorem}
 \begin{remark} The proof will show that $|I_{\Lambda_k}^{(g,g')}( \e^{ - \nu L_k })| \ge c |\Lambda_k|$.  However, for the proof of Theorem \ref{main1} we need only one $(\e^{ - \nu L_k } , \Lambda_k)$-corresponding pair.
 \end{remark}

Theorem \ref{infinite matching} is proved below in Section \ref{match}  We close this section by  utilizing the description of $\epsilon$-corresponding pairs of $H_g$ and $H_{-g}$ to analyze the 
    evolution operator  $e^{-i t \fh_g}$ on selected spin-up states and thereby complete the proof of Theorem \ref{main1}. By Proposition \ref{hgbasis},   the normalized eigenvectors of $\fh_g$ are exactly the vectors
    \be  \label{hgbasis2}
       \ket{\psi_{j}^{\pm}}   =       \frac{1}{\sqrt 2}
                          \left(   \left|\phi_{j }^{(\pm g)} ,1 \right\rangle  
                          \pm \left| \phi_{j }^{( \pm g)}  , -1\right\rangle \right)
      \ee
    for $j  \in I^{(\pm g)}$.  The evolution of certain states of the form $\ket{\psi_j^{(g)},1}$ turns out to be simple, but non-trivial, as shown in the following

     \bp \label{reflect thm} Let $j_\pm \in I^{(\pm g)}$ be the indices of an
     $\epsilon$-corresponding pair of eigenvectors. Then for all $t \ge 0$,
       \begin{equation} \label{evol1}  
        \norm{ 
      \e^{- \im t\fh_g } \ket { \phi_{j_+}^{(g)} , 1} 
        - \frac{  1 }{\sqrt 2}   
        \left(\e^{- \im t    \lambda_{j_+ }^{(g)} } \bb1
     + \e^{- \im t \lambda_{  j_- }^{(-g)} } \sigma^{(3)}
      \right) \ket{ \psi_{ j_+ }^+}   }
         \ < \ 4 \epsilon
          \end{equation}
          where $\sigma^{(3)} = \begin{psmallmatrix} 1 & 0 \\ 0 & -1 \end{psmallmatrix}$, acting on the spin coordinate.  Let $\tau = \nicefrac{\pi }{\abs{\lambda_{j+}^{(g)} -\lambda_{j-}^{(-g)}}}$.  It follows that for $n\in \bbZ$, $t \in \bbR$, and $0<\delta <1$ satisfying $\left|t - n  \tau \right| < \delta \tau $,  we have
     \be \label{evol4}  
            \norm { \e^{- \im t   \fh_g  }  \ket{  \phi_{j_+}^{(g)} ,  +1 }
             -
               \e^{- \im t \lambda_{j_+ }^{(g)}} \ket{ \phi_{j_+}^{(g)} , (-1)^n}}   
                \ < \     4( \delta +  \epsilon).
          \ee
     \ep

Let us now see how the results given so far imply the main Theorem \ref{main1}.  Theorem \ref{infinite matching} guarantees an infinite sequence of $(\e^{-\nu L_k},\Lambda_k)$-corresponding pairs of eigenvectors for $H^{(g)}$ and $H^{(-g)}$.  For each pair let $\tau_k = \nicefrac{\pi}{\lambda_{i_k}^{(g)}-\lambda_{j_k}^{(-g)}} \ge \e^{\nu L_k} $.  Applying Theorem \ref{reflect thm} to the $k$-th pair with $\epsilon = \e^{-\nu L_k}$ and $\delta = \e^{-\frac12 \nu L_k}$ yields, for large $k$,
        \[ \norm{ \e^{- \im t   \fh_g  }  \ket{  \phi_{i_k} ^{(g)} ,  1}
          -
        \e^{- \im t \lambda_{i_k}^{(g)}} \ket{ \phi_{i_k}^{(g)},(-1)^n}}
             \ < \   \e^{- \frac14 \nu L_k}  \]
     whenever $\abs{t-n\tau_k} < \e^{\frac{1}{2}\nu L_k}.$  Part 1 of Theorem \ref{main1} follows since $\frac14 |x_{i_k}^{(g)}| \le L_k$ by eq.\ \eqref{xLbound}.  The localization statement in part 2, follows from a combination of   (\ref{localC}) and (\ref{evol1}). By (\ref{evol1}), we have
          \[ \norm{ \left (\bb1 - \wh P_{x_{i_k}^{(g)}} \right )\e^{-\im t \fh_g} \ket{\phi_{i_k}^{(g)},+1}} \ \le \ \norm{\left (\bb1 - \wh P_{x_{i_k}^{(g)}}  \right ) \ket{\psi_{i_k}^+}} \ + \ 4 \e^{-\nu L_k}.\]
       By (\ref{localC}), for large $k$ we have 
        \[   \norm{\left (\bb1 - \wh P_{x_{i_k}^{(g)}}  \right ) \ket{\psi_{i_k}^+}} 
                           \ < \ e^{ - \frac12 \xi L_k   }.  \]
	Theorem \ref{main1} follows.
	
	Now we turn to the

     \begin{proof}[Proof of Proposition \ref{reflect thm}.] To begin, let us consider the expansion of $\ket{  \phi_{j_+}^{(g)} , 1}$ in the eigenbasis of $\fh_g$. 
     For $j \in I^{(-g)}$ let
         $ a_{j }  =    \diracip{    \phi^{(-g)}_{j}      }{  \phi^{( g)}_{j_+} }    $.
      Using  eq.\ (\ref{hgbasis2}) we have 
       \be \label{psij+exp}   \sigma^{(3)} \ket{\psi_{j_+}^+} \ = \ \frac1{\sqrt 2} \left(   \ket{\phi_{j_+ }^{( g)} ,+1 }  
                           - \ket{\phi_{j_+ }^{(  g)}  , - 1} \right)
                             \ = \ \sum_{j \in I^{(-g)} } a_{j}  \ket{\psi_{j}^-}   .   \ee
       Thus, we may readily represent the spin up state as
        \[ \ket{ \phi_{j_+}^{(g)} , +1 } \ = \ \frac{1}{\sqrt{2}} \left ( \ket{\psi_{j_+}^+} +\sigma^{(3)}\ket{\psi_{j_+}^+} \right ) 
                      \   =  \
                \frac{1}{\sqrt 2}   \ket{\psi_{ j_+ }^+}   
                              +
                       \frac1{\sqrt 2}\sum_{j \in I^{(- g)}    }   
                           a_{j }   \ket{\psi_{j }^{-} }.
              \]
        Applying the evolution operator to this expansion, we obtain
        \be \label{epairexp} \e^{-\im t \fh_g} \ket{ \phi_{j_+}^{(g)} , +1 } \ = \ \frac{1}{\sqrt{2}}\e^{-\im t \lambda_{j_+}^{(g)}} \ket{\psi_{j_+}^+} + \frac1{\sqrt{2}} \sum_{j\in I^{(-g)}} a_j \e^{-\im t \lambda_{j}^{(-g)}}\ket{\psi_j^-}.\ee 
        Using \eqref{psij+exp} again, we obtain
 \begin{multline*} \e^{- \im t   \fh_g   }  \ket{  \phi_{j_+}^{(g)} , +1 }  - \frac{  1 }{\sqrt 2}   
        \left(\e^{- \im t    \lambda_{j_+ }^{(g)} } \bb1
     + \e^{- \im t \lambda_{  j_- }^{(-g)} } \sigma^{(3)} 
      \right)
       \ket{ \psi_{ j_+ }^+}
                     \\  = \           \frac1{\sqrt 2}     \sum_{j\in I^{(-g)} }  a_j   
                                  \left( \e^{-\im t \lambda_{j}^{(-g)}}  - \e^{-\im t \lambda_{j_-}^{(-g)}} \right)    \ket{\psi_{j}^{-}} .     \end{multline*}
  Eq.\ \eqref{evol1} follows since 
    \begin{multline*}  \norm{ \frac1{\sqrt 2}     \sum_{j \in I^{(-g)}}    
                       a_j \left( \e^{-\im t \lambda_{j}^{(-g)}}  - \e^{-\im t \lambda_{j_-}^{(-g)}} \right)   \ket{\psi_{j}^{-}} }^2 
                \ \leq \  2  \sum_{j \in I^{(-g)}  \setminus \{j_-\}}    | a_{j}  |^2  
                  \\
                  \leq \ 2 \left ( 1 - \abs{ \diracip{\phi_{j-}^{(-g)}}{\phi_{j+}^{(g)}}}^2 \right )\
                    \leq \  \ 4 \epsilon , \end{multline*}
	by the eigenvector correspondence assumption.
	
	To prove eq.\ \eqref{evol4}, note that for any $t$ and $n$,
      \[ 
      \left( \e^{- \im t \lambda_{ j_+ }^{(g)}} \pm  \e^{- \im t  \lambda_{ j_-}^{(-g)}  } \right) 
      \ = \
        \e^{ -\im t\lambda_{ j_+}^{(g)}       } \left( 1 \pm  \e^{ \im t(\lambda_{ j_+}^{(g)} - \lambda_{ j_-}^{(-g)} ) }  \right) 
        =
      \e^{- \im t \lambda_{ j_+ }^{(g)}}
        \left( 1\pm  \e^{ \im n \pi}  \right) 
        \mp 
        \left( \e^{\im n \pi} -
        \e^{ \im t(\lambda_{ j_+}^{(g)} - \lambda_{ j_-}^{(-g)} ) }
         \right)  .\] 
        For $t $ satisfying $\left|t - n \tau\right| < \delta \tau$, we have
        \[ \abs{ \left( \e^{- \im t \lambda_{ j_+}^{(g)}} \pm  \e^{- \im t  \lambda_{ j_-}^{(-g)}  } \right) - \e^{- \im t \lambda_{ j_+ }^{(g)}}
        \left(1\pm  \e^{ \im n \pi}   \right) } \ \le \ \abs{1 - \e^{\im \pi \left ( \frac{t}{\tau} - n \right )}} \ \le \ \pi \delta.\]
     Thus, we have the approximation
     \be
     \norm{ \frac1{\sqrt{2}}
     \left( \e^{- \im t    \lambda_
                { j_+ }^{(g)} } \bb 1
      +   
      \e^{- \im t \lambda_{ j_- }^{(-g)} } \sigma^{(3)} \right)  -
     \e^{-\im t\lambda_{j_+}^{(g)}}
     \frac1{\sqrt{2}}\begin{pmatrix} 1+(-1)^n  & 0 \\ 0 & 1- (-1)^n  \end{pmatrix}
      }   \ < \ 4 \delta .
          \ee
       Combining this with (\ref{evol1})   we obtain (\ref{evol4}).
     \end{proof}

 \section{Proof Theorem \ref{infinite matching}}\label{minsmatch}
   \subsection{Concentration of localization centers}\label{indices}    
   In this section we prove Lemma \ref{localsum}. The concentration inequalities
           (\ref{locallb}) and (\ref{localub}) follow immediately from estimates (\ref{localC}) and (\ref{localD}).  We will prove (\ref{locallb}).  The proof of (\ref{localub}) is similar, but with the roles of the projections interchanged.  We have $\tr P_\Lambda^{(g)} = \tr  P_{\Lambda^{+}} P_\Lambda^{(g)}   + \tr  (1- P_{\Lambda^{+}}) P_\Lambda^{(g)}. $ Furthermore,
          \[ \tr  (1-P_{\Lambda^{+}}) P_\Lambda^{(g)} \ =\ \tr  (1- P_{\Lambda^{+}}) P_\Lambda^{(g)}P_\Lambda^{(g)} \ \leq \ \|(1-P_{\Lambda^{+}}) P_\Lambda^{(g)}\| \tr P_\Lambda^{(g)}.\]
          Thus, 
           \be\label{tracebound}
                    \left( 1 -\e^{-\nicefrac{\xi \ell}{2}}  \right)  \tr P_\Lambda^{(g)}  \ \leq \ \tr P_{\Lambda^{+}} P_\Lambda^{(g)}  \  \leq \ \tr P_{\Lambda^{+}} . 
              \ee
              Therefore, (\ref{localub}) holds provided $C_\omega$ is large enough that $- \log (1-\alpha) \le \frac{\xi}{2} C_\omega$.

          We now prove (\ref{localC}).  Let $i \in I^{(g)}_\Lambda$, a state with localization center $x =x_i \in \Lambda$.
           Then,
           \[\| (1 - P_{\Lambda^{+}}) P^{(g)}_\Lambda |\phi_i^{(g)}\rangle  \|^2
                                            \   = \ \sum_{ y  \notin  \Lambda^+}  |\phi_i^{(g)}( y )|^2
                                     \ \leq \   \sum_{y: |u - x_i | > \ell}   | \phi^{(g)}_i(y)   |^2.  \]
           We use the SULE bound 
           (\ref{suledecay}) to estimate the sum, 
         \begin{multline*} \| (1 - P_{\Lambda^{+}}) P^{(g)}_\Lambda |\phi_i^{(g)}\rangle  \|^2
                          \ \leq  \ \sum_{k = \ell}^\infty  C_d     A_\omega^2 (1 + |x_i|)^{2d+2}  k^{d  - 1}    \e^{-2 \xi  k}  
                          \\ \leq  \  \left [ C_d A_\omega^2 (1 + |u|+L)^{2d+2} \sum_{k \ge  C_\omega +\log^p ( |u|+ L) } k^{d-1} \e^{-\xi k}\right ] \e^{-\xi \ell} \ \le \   \e^{ - \xi  \ell},  \end{multline*}
          where the final inequality holds provided $ C_\omega  $ large enough.

         Now we prove (\ref{localD}). Let $x \in \Lambda^{-}$. Then 
      \[
              \|    (1 - P_{\Lambda }^{(g)}) P_{\Lambda^{-}} |x \rangle  \|^2  \ = \   \sum_{i \notin I_\Lambda^{(g)}}  |\phi_i^{(g)}(x)|^2.
       \]
       We group the sum on the right hand side by summing, for each $y\notin \Lambda$, over the eigenfunctions with localization centers at $x_i=y$,
       \[  \sum_{i \notin I_\Lambda^{(g)}}  |\phi_i^{(g)}(x)|^2 
                                 \ = \ 
                                 \sum_{y \notin \Lambda }  \sum_{i \in I^{(g)}_{ \{y\}}  } |\phi_i^{(g)}(x)|^2 .  \]  
       From (\ref{localub}) the number of states with localization center at site $y$ 
                 is bounded by  
       \[ \alpha^{-1} |\Lambda_{C_\omega + \log^p |y|  }(y)| \leq  \alpha^{-1} C_d  \left ( C_\omega + \log^p |y|\right )^{d}  .\]     
       Thus, using the SULE bound \eqref{suledecay}, we have,
         \[
                            \sum_{i \notin I^{(g) }_\Lambda    }   | \phi^{(g)}_i(x)  |^2     
                       \   \leq \  \alpha^{-1} C_d A_\omega^2
                           \sum_{y : |y -x| \geq  \ell  } 
               \left (C_\omega + \log^p |y| \right )^{d}    
                   (1 + |y|)^{2 d + 2}  \e^{-2 \xi |x - y| }   . 
            \]
        Recalling that $x\in \Lambda^- = \Lambda_{L-\ell}(u)$, we see that $|y| \le |y-x| + L + |u|$, so that 
        \[ \sum_{i \notin I^{(g) }_\Lambda    }   | \phi^{(g)}_i(x)  |^2     
                       \   \leq \ \alpha^{-1} C_d A_\omega^2
                           \sum_{k \ge C_\omega + \log^p(L+|u|)}^\infty 
               \left ( C_\omega + \log^p ( L + |u| +k) \right )^{d}    
                   (1 + L + |u| +k)^{2 d + 2}  \e^{-2 \xi k}. \]
        Clearly we may choose $C_\omega$ large enough that  $ \sum_{i \notin I^{(g) }_\Lambda    }   | \phi^{(g)}_i(x)  |^2  <  \e^{- \xi \ell } $, which implies (\ref{localD}).
This completes the proof of the lemma.

   \subsection{Minami Estimate}\label{mins}
   The Minami estimate\cite{M1996} for the Anderson model $H = \gamma \Delta + V$ is an \emph{a priori} bound on the eigenvalue correlations for $H$ restricted to a box $\Lambda$. For an operator $A$ on $\ell^2(\bbZ^d)$ and a subset $\Lambda\subset \bbZ^d$, let $A_\Lambda = P_\Lambda A P_\Lambda $, the restriction of $A$ to $\Lambda$ ``with Dirichlet boundary conditions.'' If $\Lambda \cap \supp \zeta = \emptyset$, we have  $(H+gD)_\Lambda = H_\Lambda$. As we will consider only boxes disjoint from the support of  $\zeta$,
       we will simply use $H_\Lambda$ below without further comment.
       
       Minami's estimate\cite{M1996} is the following
          \begin{theorem}\label{minami}
       For any interval $J \subset \bbR$ and subset $ \Lambda \subset \bbZ^d $,
      \[   \bbP\left( \tr P_{H_\Lambda  }(J) \geq 2  \right) 
               \ \leq \ \frac{\pi^2}{2} \left(  \|\rho\|_\infty | J |   |\Lambda|  \right)^2. \]
   \end{theorem}
   
   As an immediate corollary we find a probabilistic  bound on the minimal separation between eigenvalues.  For a finite set $T=  \{t_1,..,t_N\}$, let
             \[  \Delta_{\min} [T] \ := \ \min\{  |t_i - t_j| :   i\neq j    \}. \] 
   Then, Minimi's estimate implies
   \begin{corollary}
         For any $\rho $ with compact support,
       there is a finite  $C$ so that for any $\epsilon > 0$,
       \[        \bbP\left(    \Delta_{\min}[ \sigma (H_\Lambda) ] < \epsilon  \right)     
          \ < \ C    |\Lambda|^2 \epsilon.  \]   
     \end{corollary}
     \begin{remark} To prove this corollary, note that the spectrum may be covered by $O(\epsilon^{-1})$ intervals of length $\epsilon$ in such a way that if $\Delta_{\min}[ \sigma (H_\Lambda) ] <\epsilon$ then at least one of these intervals contains at least two eigenvalues.\end{remark}
   
   By a Borel-Cantelli argument we can control the minimum spacing for a sequence of volumes. 

   \begin{corollary}\label{minami2} Let $(u_k)_{k= 1}^\infty $ be a sequence of sites $u_k\in \bbZ^d$ and
    $(L_k)_{k = 1}^\infty$ a sequence of length scales  with $L_k \to \infty$.  Suppose that $\Lambda_k \cap \Lambda_{k'}=\emptyset$ for $k\neq k'$, where     $\Lambda_k  = \Lambda_{L_k} (u_k) $. Then there is a set $\Omega_{BC} \subset \Omega$ of measure 1 so that, for infinitely many $k \in \bbN$ 
         \be \label{BCsep}
          \Delta_{\min}[\sigma (H_{\Lambda_k}) ] \ > \ \frac{1}{L_k^{2d +1}}   .
         \ee
   \end{corollary}
     \begin{proof}
        Let $\epsilon_k= L_k^{-2d-1}$. The events $\{ \Delta_{\min}[\sigma (H_{\Lambda_k})]  > \epsilon_k \}$ are independent,
         and obey
           \[  \bbP\left(   \Delta_{\min}[\sigma (H_{\Lambda_k})]  > \epsilon_k  \right)  > 1 -   C L_k^{ - 1},    \]
               In particular,   $\bbP\left(   \Delta_{\min}[\sigma (H_{\Lambda_k})] \right)  \to 1 $, 
            so that the sum of the probabilities over $k$ is infinite. 
        Thus,   the corollary follows from the second Borel-Cantelli theorem.
    \end{proof}

   \subsection{Matching Eigenbases}\label{match}

	We  now combine the Minami estimate with SULE localization bounds to prove that eigenvalue labelings are stable under  rank one perturbations and prove Theorem \ref{infinite matching}. We will use the following lemma in the proof.  Recall that $\Sigma_{\Lambda}^{(g)}$ denotes the local spectrum of $H_g$, namely the set of eigenvalues $\lambda_i^{(g)}$ with  corresponding localization center $x_i^{(g)} \in \Lambda$. 
	
	In the following, we will work with the sequence $(u_k)_{k=0}^\infty$ and $(L_k)_{k=0}^\infty$ defined in the paragraph preceding Theorem \ref{infinite matching}.  For fixed small $\beta > 0$, let us define 
  \[ \Lambda_k^\pm \ = \ \Lambda_{(1 \pm \beta) L_k }(u_k), \] 
analogous to $\Lambda^\pm  $ in Lemma \ref{localsum}.
       \begin{lemma} \label{matching} Let $\mc{G} \subset \bbR$ be a finite set. For all $\omega \in \Omega_{\mathrm{SULE}}^{\mc{G}} = \cap_{g\in \mc{G}} \Omega_{\mathrm{SULE}}^{(g)}$ there is a finite $k_\omega$ so that for all $k > k_\omega$ the following holds:
        \begin{enumerate}
            \item 
               For every $k > k_\omega$ and $g \in \mc{G}$, there is a map
       \[
         \Psi_k^{(g)} : \Sigma_{\Lambda_k}^{(g)} \ \rightarrow \ \sigma( H_{\Lambda_k^+} ) \]                so that  $ | \lambda_i^{(g)} - \Psi_k^{(g)}(\lambda_i^{(g)})    | < \e^{- \frac14 \xi \beta L_k  }.$         
        \end{enumerate}
    Moreover,  if $ H_{ \Lambda^+_k }$ has simple spectrum satisfying $\Delta_{\min}[\sigma ( H_{ \Lambda_k^+  } ) ]>  L_k^{- 2d-1}$, then the following hold for large enough $k$:
            \begin{enumerate}[resume]
    \item
     For $g',g''\in \mc{G}$, suppose there are $i\in I^{(g')}_{\Lambda_k}$ and $j \in I^{(g'')}_{\Lambda_k}$ so that,
   \be\label{psimatch}
      \Psi_k^{(g')}(\lambda_i^{(g')})  
        =  \Psi_k^{(g'')}(\lambda_j^{(g'')}),   \ee 
    then the eigenvectors $\phi_i^{(g)}, \phi_j^{(g'')} $ are $\e^{-\frac14 \xi \beta L_k} $-corresponding.
     \item       For any  $g \in \mc{G}$,  the map $\Psi_k^{(g)}$ is one-to-one.
     \item For any $g\in \mc{G}$, there is a minimum separation for the local spectrum at $\Lambda_k$:  
     \[    \Delta_{\min}  \left[ \Sigma_{\Lambda_k}^{(g)} \right] > \frac12
                L_k^{- 2d-1}. \]
   \item For any $\alpha < 1 $, large enough $k$ and  $g',g'' \in \mc{G}$,
         \be \label{matchedpairs} 
       |I^{(g',g'')}_{\Lambda_k}(  \e^{- \frac14 \xi \beta L_k  } )| \geq 2\alpha |\Lambda_k^-| - |\Lambda_k^+| .
        \ee
              \end{enumerate} 
      \end{lemma}

     Before proving the lemma, let us show how it implies Theorem \ref{infinite matching}.
   
      \begin{proof}[Proof of Theorem \ref{infinite matching}]  Let $\epsilon_k =   L_k^{ -2 d -1} $.
      Then Corollary \ref{minami2} implies there is a measure one set $\Omega_{BC} \subset \Omega $ so that for all $\omega \in \Omega_{BC}$ there are infinitely many $k \in \bbN $ for which
              $\Delta_{\min}[\sigma (H_{\Lambda_k}) ] > \epsilon_k$.
      Let $k_\omega$ be as defined in Lemma \ref{matching}, and for any  $\omega \in \Omega_{BC } \cap \Omega_{\mathrm{SULE}} $, 
           let
  \[ \cK_\omega \ = \ \setb{  k \geqq k_\omega}{\Delta_{\min}[\sigma (H_{\Lambda_k}) ] > \epsilon_k }.  \]
          We apply Lemma \ref{matching}, for sufficiently large $k \in \cK_\omega$. From part 5, there are \[ 2 \alpha|\Lambda^-_k| - |\Lambda^+_k| \ = \ C_d \left ( 2 \alpha (1-\beta)^{d} - (1+\beta)^d \right ) L_k^d \geq \  \frac{1}{2} |\Lambda^-_k|  \ > \ 0 \] many $\e^{-  \frac14 \xi \beta L_k }$-corresponding pairs of eigenpairs for a suitable choice of $\alpha<1$ and $\beta>0$. The result follows with $\epsilon = \frac14 \xi \beta$.   
     \end{proof}

   Finally we prove Lemma \ref{matching}
     \begin{proof}[Proof of Lemma \ref{matching}]  Fix $1<p<2$, $\alpha <1$, and let $C_\omega =\max_{g\in \mc{G}} C_\omega(p,\alpha,g)$, with $C_\omega(g)$ as in Lemma \ref{localsum}. Let $k_\omega$ be chosen so that  $\beta L_k \ge C_\omega + \log^p L_k^+$ for $k > k_\omega$.
     Then for $g\in \mc{G}$ and $i \in I^{(g)}_{\Lambda_k}$, we have       $  \|(1 - P_{ \Lambda_k^+ }) 
         \phi_i^{(g)} \|  <  \e^{-\frac12 \xi  \beta L_k } ,$ by  (\ref{localC}). Increasing $k_\omega$ if necessary, we shall assume that $ \e^{-\frac12 \xi \beta L_k} < \frac{1}{\max(2,4\gamma^2 d)}$ for $k>k_\omega$. 
        
        Thus for $g\in \mc{G}$ and $i\in I^{(g)}_{\Lambda_k}$ we have 
        \[ \norm{P_{\Lambda_k^+} \phi_i^{(g)}} \ > \ \frac12 .  \]
        Let
        \[ (\phi_i^{(g)})_{ \Lambda_k^+ } \ := \ \frac{1}{\norm{P_{\Lambda_k^+} \phi_i^{(g)}} } P_{\Lambda_k^+} \phi_i^{(g)} \]
        denote the normalized restriction of $\phi_i^{(g)}$ to $\Lambda_k^+$ and set
     	\[  R^{(g)}_i \ := \  \left( H_{\Lambda_k^+ }  
                        - \lambda^{(g)}_i \right) (\phi_i^{(g)})_{ \Lambda_k^+ }
             \ = \   - \frac{\gamma}{\|P_{ \Lambda_k^+ } \phi_i^{(g)}\|} 
              \sum_{x \in \partial \Lambda_k^+ } | x \rangle \sum_{y \sim x: y \notin \Lambda_k^+ } \phi_i^{(g)}(y) .  \]
        We again use 
        (\ref{localC}), to bound the norm of the remainder,
           \[  \|R^{(g)}_i\|^2
             \leq 4 d \gamma^2 \| (1 - P_{\Lambda_k^+})  \phi_i^{(g)}\|^2
                \ \leq \  4 \gamma^2 d \e^{- \xi \beta L_k } \ \leq \ \e^{- \half \xi \beta L_k }. \]
         On the other hand, observe that
           \be \|R^{(g)}_i\|^2  \ = \ \abs{\dirac{ (\phi_i^{(g)})_{ \Lambda_k^+ } }{    (H_{ \Lambda_k^+ } - \lambda_i^{(g)})^2  }{  (\phi_i^{(g)})_{ \Lambda_k^+ } }}.
              \label{finiteboxapprx}  
              \ee     
      By the min-max theorem applied to $(H_{ \Lambda_k^+ } - \lambda_i^{(g)})^2 $, we see that there is an eigenvalue $\eta \in \sigma(H_{ \Lambda_k^+ })$, with an associated eigenvector $\phi \in \ell^2( \Lambda_k^+)$,    so that
      $|\eta - \lambda_i^{(g)}| < \e^{- \frac14 \xi \beta L_k  }  $. This shows the existence of $\Psi_k^{(g)}$, and proves part 1.

   With $\eta$  as above,  enumerate the eigenvalues of $(H_{ \Lambda_k^+ } - \eta)^2$ as $0 = \eta_0\leq \eta_1 \leq \cdots \leq \eta_N$. Let $\overline{P}_j$ project to the eigenspace of $(H_{ \Lambda_k^+ } - \eta)^2$ associated to $\eta_j$.
             From the calculation in part 1 we have,    
            \be \label{match2e1}
                   \dirac{    (\phi_i^{(g)})_{ \Lambda_k^+ }  }{ (H_{\Lambda_k^+ } - \eta)^2 }{(\phi_i^{(g)})_{\Lambda_k^+ }  }
             \ = \  \norm{   (  \lambda_i^{(g)} - \eta ) (\phi_i^{(g)})_{\Lambda_k^+ } + R_i^{(g)} }^2  \ \leq \  4 e^{ -\frac12 \xi \beta L_k } .   
              \ee
     On the other hand we have,
          \begin{multline}\label{match2e2}
          \dirac{    (\phi_i^{(g)})_{ \Lambda_k^+ }  }{ (H_{\Lambda_k^+ } - \eta)^2 }{(\phi_i^{(g)})_{\Lambda_k^+ }  }    
        \ =  \  \sum_{j = 1}^\infty \eta_j   \dirac{     (\phi_i)_{\Lambda_k^+ }  }{ \overline{P}_j }{(\phi_i)_{\Lambda_k^+ } } \\
        \geq  \         \eta_1 \dirac{   (\phi_i^{(g)})_{\Lambda_k^+ } }{ ( 1 -  \overline{P}_0 ) }{(\phi_i^{(g)})_{\Lambda_k^+ }} . 
              \end{multline}
     Combining (\ref{match2e1}) and (\ref{match2e2}), we obtain 
            \be \label{lowerbound1}
            \dirac{   (\phi_i^{(g)})_{\Lambda_k^+ } }{ \overline{P}_0 }{(\phi_i^{(g)})_{\Lambda_k^+ }} \ \geq \ 
   1 - \frac{ 4 \e^{ - \xi \beta L_k/ 2}}{\eta_1}.  
              \ee  
     
     We will obtain both conclusions 2 and 3 of the lemma from equation (\ref{lowerbound1}).
     Suppose that    $\sigma ( H_{\Lambda_k^+ } ) $ is simple and   $  \Delta_{\min} [  \sigma ( H_{\Lambda_k^+ } ) ]> L_k^{- 2d -1} $. Under this assumption, $\eta_1 > L_k^{- 4d -2} $ and $\overline{P}_0 = \ket{\phi}\bra{\phi}$ where $\phi$ is the eigenvector for $H_{\Lambda_k^+}$ with eigenvalue $\eta$. 
     
     Suppose that there are  indices $i\in I^{(g')}_{\Lambda_k}  $ and $j \in I^{(g'')}_{\Lambda_k}$  so that (\ref{psimatch}) holds. Let $ \eta = \Psi_k^{(g')}(\lambda_i^{(g')}) = \Psi_k^{(g'')}(\lambda_j^{(g'')})  $ and let $\phi $ be the eigenvector associated to $ \eta$ in $ H_{\Lambda_k}$. By (\ref{lowerbound1}), for $ ( l , g) = (i,g')$ or $=(j,g'') $ we have     \[
   \ket{(\phi_l^{(g) })_{\Lambda_k^+ }}
        \ = \ \diracip{ \phi }{ (\phi_l^{(g) })_{\Lambda_k^+ }}    \ \ket{ \phi } + \ket{  \hat \phi_l}
     \] 
           where $\|\hat \phi_l\|^2 \leq    \frac{4}{\eta_1}  \e^{-  \frac12\xi \beta L_k}  $ and $\langle \phi | \hat  \phi_l \rangle = 0$.
    It follows that 
            \[
            	 \left| \diracip{( \phi_i^{(g') } )_{\Lambda_k^+} }{ (\phi_j^{(g'')} )_{\Lambda_k^+ }} \right| \
                       \geq \
             \left|  \diracip{ \phi }{ (\phi_i^{(g') })_{\Lambda_k}}
                     \diracip{ \phi }{ (\phi_j^{(g'') })_{\Lambda_k}}  \right|
                                       -   \abs{ \diracip{ \hat \phi_i }{ \hat  \phi_j}}
             \          \geq   \
                                     1 - 8 L_k^{4d+2}  e^{- \frac12 \xi \beta L_k }.  
            \]
     Thus the pair $ \phi_{ i}^{( g')}, \phi_{ j}^{(g'')}  $ is $8 L_k^{4d+ 2} e^{-\frac12 \xi \beta L_k}$-corresponding so that part 2 holds.
          
    Applying part 2 with $g'=g''$, we see that if we were to have $\Psi_k^{(g)}(\lambda_i^{(g)}) = \Psi_k^{(g)}(\lambda_j^{(g)})$ for $i\neq j$ then it would happen that
    \[ \diracip{\phi_i^{(g)}}{\phi_j^{(g)}} \ \ge \ 1- \e^{-\frac14 \xi \beta L_k}, \]
    contradicting the fact that $\diracip{\phi_i^{(g)}}{\phi_j^{(g)}}=0$.  Thus $\Psi_k^{(g)}$ is injective for large $k$ and part 3 holds.
          The minimum separation claimed in part 4 follows since   
          \begin{multline*}
          	| \lambda_i^{(g)} -   \lambda_j^{(g)}  |  \ \ge \ \abs{\Psi_k^{(g)}(\lambda_i^{(g)})-\Psi_k^{(g)}(\lambda_j^{(g)})} - \abs{\lambda_i^{(g)} -\Psi_k^{(g)}(\lambda_i^{(g)})} - \abs{\lambda_j^{(g)} -\Psi_k^{(g)}(\lambda_j^{(g)})}
          	\\ \ge \   L_k^{- 2d -1} - 2 \e^{ -\frac14 \xi \beta L_k },
          \end{multline*}  
           by parts 1 and 3 and the minimum separation for $\sigma(H_{\Lambda_k^+} )$.      
		 
		 Finally, we prove part 5. There are $  |I_{\Lambda_k}^{(g)} |$ elements in the range of $ \Psi_k^{(g)} $ for $g = g',g''$. Since $\Psi^{(g)}_k$ is an injective map into $ \sigma(H_{\Lambda_k^+}) $, 
               \[ 
               |  \rng \Psi_k^{(g')} \cap \rng\Psi_k^{(g'')} | \ \geq \ |I_{\Lambda_k}^{(g')} | +  |I_{\Lambda_k}^{(g'')} |- |\sigma(H_{\Lambda_k^+})| \ \ge \ 2\alpha |\Lambda_k^-| - |\Lambda_k^+| ,
                \]
         by \eqref{locallb}. To complete the proof, note that part 2 implies that every element in $ \rng \Psi_k^{(g')} \cap \rng\Psi_k^{(g'')} $ corresponds to a    $ \e^{-\frac14 \xi \beta L_k}$-corresponding pair of eigenvectors.    
     \end{proof}

   \appendix

    \section{Fractional moments and Dynamical localization}\label{app1}
     In this section, we recall the fractional moment bounds for the Anderson model
     and the related dynamical localization results. Large disorder localization for the Green's function follows from the following
    \begin{theorem}[See ref. \onlinecite{AW2015} (Theorem 6.3)]\label{fmsaw}
        There is a finite constant $C = C_{\rho, d}$ so that for any $ 0 < s < 1$ and any $z \in \bbC \setminus \bbR$
        \[  \bbE (  | \langle x |   (H - z)^{-1}  | y \rangle |^s ) \   \leq \ \frac{C}{ 1- s}  \left( \frac{\gamma^s C}{ 1 - s} \right)^{|x - y|}  \]
    \end{theorem}
    For the Anderson model,  transition probabilities may be bounded in terms  of fractional moments. 
    \begin{theorem}[See ref.\ \onlinecite{AW2015} (Theorem 7.7)] \label{fm db 2}
        There is $C   < \infty$ depending on $s$ so that,
     \[         \bbE ( | \langle x | e^{ -  it H }  | y \rangle |   )   
             \   \leq  \  C   \liminf_{\epsilon \searrow 0}   \int_I   \bbE (  | \langle x |   (H - ( E + i\epsilon ) )^{-1}  | y \rangle |^s ) dE    \]
     \end{theorem}

     From these statements we can easily demonstrate Theorem \ref{fractional moment}
   and Corollary \ref{fm dynamical bound}.
    \begin{proof}[Proof of Theorem \ref{fractional moment}]
    We will only prove the result for sufficiently small $s$.  As explained above, the result follows for all $s <1$ by an ``all-for-one'' lemma such as Lemma B.4 of ref. \onlinecite{Aizenman2001}.    Furthermore, we note that the \emph{a priori} bound
    \be \label{apriori} \Ev{ \abs{\dirac{x,i}{(\fh_g -z)^{-1}}{y,i}}}^s \ \le \ \frac{C}{1-s} \ee 
    holds in this context \tem \ for example, see Lemma 2.5 of ref.\ \onlinecite{MS2016}. 
          
    Let  $ \fD = \ket{\zeta}\bra{\zeta}\otimes \sigma^{(1)}$.
         By the resolvent equation 
        \be \label{determinant}
             (\fh_g - z)^{-1} \ = \ ( \fh_0  - z)^{-1} -  g (\fh_0 - z )^{-1} \fD  (\fh_g  -  z)^{-1},
          \ee
        for $z \in \bbC \setminus \bbR$.
        Two applications of of eq.\ \eqref{determinant} imply
       \begin{multline*}    \dirac{ x,  i }{ (\fh_g - z)^{-1}}{ y ,i} \ = \
                   \dirac{ x}{  ( H  - z)^{-1}}{  y } \\+ 
                  g^2  \langle x  |(H - z )^{-1} |\zeta\rangle
                    \langle \zeta, -i | (\fh_g  -  z)^{-1} |  \zeta,i\rangle   \langle \zeta  |(H - z )^{-1} | y \rangle.
        \end{multline*}
       
       Now we take the $s^{\mathrm{th}}$ moment for $0 < s < \nicefrac{1}{3}$. For the first term we may apply 
       Theorem \ref{fmsaw} directly. For the second, we apply H\"older's inequality to find,
       \begin{align*}     \Ev{|\langle x,  i | (\fh_g - z)^{-1} |  y ,i\rangle|^s }  &  \ \leq \ 
                                  \Ev{|\langle x  | (H - z)^{-1} |  y \rangle|^s}  + \\
               &   C_s g^{2s}  \left [ \Ev{ | \langle x  |(H - z )^{-1} |\zeta\rangle|^{3s}} 
                        (\Ev{ |\langle \zeta  |(H - z )^{-1} | y \rangle|^{3s }} \right ]^{\nicefrac{1}{3}},
        \end{align*}
        where we have used the \emph{a priori} bound \eqref{apriori} to estimate the middle factor. Then we have, using Theorem \ref{fmsaw},
       \begin{align*}  
          \Ev{ |\langle x,  i | (\fh_g - z)^{-1} |  y,i \rangle |^s}
             &  \ \leq  \ C_s \left( \frac{\gamma^s C}{ 1 - s} \right)^{|x - y|} + 
                 g^{2s} C_s
                     \sum_{u,u' \in \supp\zeta} \left( \frac{\gamma^{3s} C}{ 1 - 3s} \right)^{\frac{1}{3} \left ( |x - u|+|y-u'| \right )} \\
              &  \leq  \ \wt C_s   \left( \frac{\gamma^{3s} C}{ 1 - 3s} \right)^{\frac{1}{3} |x - y|},
        \end{align*} 
      concluding the result in this case.

        For the case $j=-i$, note that 
\[        \abs{\dirac{x,i}{(\fh_g - z)^{-1} }{y,j}} \ = \ 
				 \abs{\dirac{y,j}{(\fh_g - z)^{-1} }{x,i}} \ = \
				  \abs{\dirac{x,-i}{(\fh_g - z)^{-1} }{y,-j}}.\]
Thus we may assume without loss of generality that $i=1$, $j=-1$ and $|x|\ge |y|$.
				  Eq.\ (\ref{determinant}) implies
          \[       
              \langle x, 1 | (\fh_g - z)^{-1} |  y , -1 \rangle \ = \ 
               -  g  \langle x |  (H - z )^{-1}  |\zeta   \rangle
        \langle \zeta , -1  | (\fh_g  -  z)^{-1} |  y, -1\rangle
          \]
      Taking the $s$-moment for  $ 0  < s < \nicefrac{1}{2} $ and applying H\"older's inequality yields
      \[
          \Ev{ |\langle x, 1 | (\fh_g - z)^{-1} |  y,-1 \rangle |^s }
               \ \leq \     C_s g^s
                    \left[  \Ev{ | \langle x,1 |  (H - z )^{-1}  |\zeta , 1  \rangle |^{2s  } }\right]^{1/2}
        \] 
        where we have again used the \emph{a priori} bound \eqref{apriori} to estimate the second factor. Since $|x|\ge |y|$, we obtain        
        \begin{align*}  \Ev{\abs{ \dirac{x,1}{(\fh_g - z)^{-1} }{y,2}}^s}
              \ \leq  \ 
                 g^{s}C_s                     \sum_{u \in \supp\zeta}    \left( \frac{\gamma^s C}{ 1 - s} \right)^{|x - u|}  
                \ \leq \ \wt{C}_s   \left( \frac{\gamma^s C}{ 1 - s} \right)^{\frac{1}{2}(|x| +|y|) },
        \end{align*}  
         which concludes the proof of Theorem \ref{fractional moment}  
      \end{proof}
  Now let us prove Corollary \ref{fm dynamical bound}.
    \begin{proof}[Proof of Corollary \ref{fm dynamical bound}.]
        By the perturbation formula for semigroups,
         \[    
                \e^{- \im t\fh_g}  = \e^{-\im t\fh}   - \im g \int_0^t  \e^{- \im (t-s) \fh_g } \fD \e^{ - \im s \fh_0  } ds  .     \]      Thus,
         \begin{multline*}   
                 \Ev{  \abs{\dirac{ y, - 1 }{ \e^{- \im  t\fh_g} }{x , +1\rangle }} }
                                  \leq \ 
			g \int_0^t \Ev{ \abs{ \dirac{ y, - 1  }{ \e^{- \im  (t-s) \fh_g } }{ \zeta, - 1 }}
					\abs{\dirac{ \zeta, + 1 }{  \e^{ - \im s \fh_0  }  }{  x , +1}}}ds  
                                \\
                                \leq  \ 
                              g  t  \ \sup_s   \Ev{ \abs{  \dirac{  \zeta }{  \e^{ - \im s H  }  }{ x  } }}.
      \end{multline*}
       The corollary now follows from  Theorems \ref{fmsaw} and \ref{fm db 2}.
    \end{proof}

\end{document}